\documentclass[envcountsame]{llncs}
\usepackage{amsmath}
\usepackage{amssymb}

\newcommand{\IK}{{\cal I \cal K}}
\newcommand{\IIK}{{\cal I \cal I \cal K}}
\newcommand{\bmath}[1]{ {\boldmath $ {#1}$}}
\newcommand{\argmax}{\mathop{\rm argmax}}
\newcommand{\conv}{\mathop{\rm conv}}

\usepackage{datetime}

\newcommand{\R}{\mathbb{R}}

\begin{document}

\title{Approximation Algorithms for the Incremental Knapsack Problem via Disjunctive Programming}
\author{
        Daniel Bienstock, Jay Sethuraman, Chun Ye 
       }
\institute{Department of Industrial Engineering and Operations Research\\
        Columbia University\\}
\date{\today}
\maketitle

\begin{abstract}
In the \emph{incremental knapsack problem} ($\IK$), we are given a knapsack whose capacity grows weakly as a function of time. There is a time horizon of $T$ periods and the capacity of the knapsack is $B_t$ in period $t$ for $t = 1, \ldots, T$. We are also given a set $S$ of $N$ items to be placed in the knapsack. Item $i$ has a value of $v_i$ and a weight of $w_i$ that is independent of the time period. At any time period $t$, the sum of the weights of the items in the knapsack cannot exceed the knapsack capacity $B_t$. Moreover, once an item is placed in the knapsack, it cannot be removed from the knapsack at a later time period. We seek to maximize the sum of (discounted) knapsack values over time subject to the capacity constraints. We first give a constant factor approximation algorithm for $\IK$,  under mild restrictions on the growth rate of $B_t$ (the constant factor depends on the growth rate). We then give a PTAS for $\IIK$, the special case of $\IK$ with no discounting, when $T = O(\sqrt{\log N})$. 
\end{abstract}
\newpage

\section{Introduction}
Traditional optimization problems often deal with a setting where the input parameters to the optimization problem are static. However, the static solution that we obtain may be inadequate for a system whose parameters---the inputs to our optimization problem---change over time. We consider one special case of this dynamic environment in which we have a maximization problem subject to certain capacity constraints. All of the inputs to the optimization problem are static except the capacities, which increase weakly over time. The goal is to find a sequence of compatible feasible solutions over time that maximizes a certain aggregate objective function. We will call such an optimization problem \emph{an incremental optimization problem}. Unlike online and stochastic optimization problems, there is no uncertainty in the input parameters for the optimization.

In this paper we consider the {\em incremental knapsack problem}, which is a particular case of the incremental optimization problem. In the \emph{discrete} incremental knapsack problem, we are given a knapsack whose capacity grows as a function of time. There is a time horizon of $T$ periods and the capacity of the knapsack is $B_t$ in period $t$ for $t = 1, \ldots, T$. We are also given a set $S$ of $N$ items to be placed in the knapsack. Item $i$ has a weight $w_i > 0$ that is independent of the time period, and a value at time $t$ of the form $v_i \Delta_t$ where $v_i > 0$ and $\Delta_t > 0$ (this particular
functional form will allow us to model discounting). At any time period $t$, the sum of the weights of the items in the knapsack cannot exceed the knapsack capacity $B_t$. Moreover, once an item is placed in the knapsack, it cannot be removed from the knapsack at a later time period. We are interested in maximizing the sum over the $T$ time periods of the total value of the knapsack in each time period.  

To put it formally, for $X \subseteq S$ define $V(X)$ to be $\sum_{i \in X} v_i$ and $W(X)$ to be $\sum_{i \in X} w_i$. Then we are interested in finding a feasible solution $F = \{S_1, S_2, \ldots, S_T\}$ and $S_1 \subseteq S_2 \ldots, S_T \subseteq S$, where $S_t$ represents the subset of items in the knapsack in period $t$, that maximizes the quantity $\sum_{t = 1}^{T} V(S_t) \Delta_t$ subject to the constraints $W(S_t) \leq B_t$ for $t = 1, \ldots, T$.   The special case where $\Delta_t = 1$ for all $t$ will be called \textit{time-invariant}.
For brevity, in what follows we will denote the incremental knapsack problem as \bmath{\IK}, and its time-invariant version as \bmath{\IIK}.

One can also consider a continuous version of the problem.  Here we assume that we have a continuous parameter time 
parameter $s \in [0, S]$ for some $S > 0$. We are given a knapsack capacity function $B(s)$, weakly increasing with respect to $s$, and a set $K$ of $N$ items to be placed in the knapsack. Item $i$ has a value of $v_i$ and a weight of $w_i$, both time independent. At any time $s$, the sum of the weights of the items in the knapsack cannot exceed the knapsack capacity $B(s)$. Moreover, once an item is placed in the knapsack, it cannot be removed from the knapsack at a later time. We are interested in finding a feasible solution $F = \{K(s)\}_{s \in [0,S]}$ that maximizes the quantity $\int_{s = 1}^{S} V(K(s))ds$,
where $V(K(s))$ is the total value of the items found in the knapsack at time $s$, under $F$.  This problem can be approximated by partitioning $[0,S]$ into a finite set of intervals of length $\Delta_t$, $t = 1, \ldots, T$.  Under mild assumptions on
the capacity function $B(s)$, the approximation provided by this discretization can be made arbitrarily tight. 

As the single period knapsack problem is already known to be NP-hard, we consider polynomial time approximation algorithms for $\IK$. For a maximization problem, a \emph{$k$-approximation algorithm} (for some $k \leq 1$) is a polynomial time algorithm that guarantees, for all instances of the problem, a solution whose value is within $k$ times the value of an optimal solution. Moreover, we say that the maximization problem has a (fully) polynomial time approximation scheme, or a PTAS (FPTAS respectively), if for every $0 \leq \epsilon < 1$, the algorithm guarantees, for all instances of the problem, a solution whose value is within $1-\epsilon$ times the value of an optimal solution. Moreover, the algorithm runs in polynomial time in the size of the inputs (and $\epsilon$) for every fixed $\epsilon$.

\subsection{Related Work}

The special case of $\IIK$ where $v_i = w_i$  for all $i$ 
has been examined in the literature. This problem is known as the \emph{incremental subset sum problem}. Hartline \cite{Hartline08}  gave a $1/2$-approximation algorithm for the incremental subset sum problem via dynamic programming. Sharp \cite{Sharp07} gave a PTAS for the incremental subset sum problem 
that applies when $T$ is a constant. This algorithm uses a variant of the dynamic programming algorithm for the standard (i.e., 1-period) knapsack problem, and it runs in time $O((\frac{VN}{\epsilon})^{T})$, where $V = \max_{i}\{ v_i\}$. 

Further, it can be shown that the incremental subset sum is strongly NP-hard by a reduction from the $3$-partition problem (proof provided in the Appendix). 

\begin{proposition} \label{Strongly NP-hard}
The incremental subset sum problem is strongly NP hard.
\end{proposition}

Consequently, the classic result of Garey and Johnson \cite{Garey1990} rules out an FPTAS for the incremental subset sum problem (and hence for $\IIK$)  both unless $P=NP$. 

A well-studied problem related to $\IIK$  is the \emph{generalized assignment problem} (GAP). In the 
generalized assignment problem, we are given a set of $m$ knapsacks and $N$ items, with knapsack $j$ having a capacity $b_j$. Further, placing item $i$ in knapsack $j$ consumes $w_{ij}$ units of capacity of  knapsack $j$, and generates a value of $v_{ij}$. 
Notice that a variant of  $\IIK$ where one is only allowed to pack, at each time $t$, an additional $B_{t+1} - B_{t}$ units, is a special case of the generalized assignment problem:
Here, we would set $b_t = B_{t+1} - B_{t}$ and $w_{it} = w_i$ for all $i$ and $v_{it} = (T-t+1)v_i$ for all $i$ and $t$. However, 
$\IIK$ is not a special case of GAP because in $\IIK$ we are allowed to pack more than
$B_{t+1} - B_{t}$ units at time $t$, assuming the knapsack has spare capacity from earlier time periods.
Approximation algorithms for the generalized assignment problem have been studied by \cite{FleischerEtAl06},\cite{ShmoysT93}. The best known constant factor algorithm is due to  Fleischer et. al. \cite{FleischerEtAl06} with  an approximation ratio of $(1- 1/e - \epsilon)$. They also showed that no approximation algorithm can beat the lower bound of $(1- 1/e)$ unless $NP \subseteq DTIME(n^{O(\log \log(n)})$. Unfortunately, these results are not directly applicable to $\IIK$, because the knapsack capacities cannot be decomposed over time.

A special case of the generalized assignment problem where the items' weight and value are identical across knapsacks is known as the \emph{multiple knapsack problem (MKP)}; for this problem, Chekuri and Khanna \cite{ChekuriK05} developed a PTAS. Moreover, they also showed that two mild generalization of the MKP---
$w_{ij} \in \{w_{i1}, w_{i2}\}$ and $v_{ij} = v_i$ or $v_{ij} \in \{v_{i1}, v_{i2}\}$ and $w_{ij} = w_i$---
are APX hard, thus ruling out a PTAS for these generalizations, assuming $P \not = NP$. Again, neither the PTAS nor their hardness results are directly applicable to the $\IIK$.

\subsection{Our Contributions}
Our first result is a constant-factor approximation algorithm for $\IK$ under some mild assumptions on the growth rate of $B_t$. This algorithm rounds the solution to a polynomial-size linear programming relaxation to the problem, specifically,
a disjunctive formulation (background and details, below). It is worth noting that, as is shown in Section \ref{LP relax}, a standard formulation for the time-invariant incremental knapsack problem has an unbounded integrality gap---this is unlike the case for the standard knapsack problem. To the best of our knowledge, no constant factor approximation algorithm is known for $\IK$ before this work. The previous best algorithm is a general purpose approximation algorithm for incremental optimization problems due to Hartline and Sharp \cite{HartlineS06}, which yield a $O(1/\log T)$ approximation ratio. 

Our second result provides a PTAS for $\IIK$ and when $T = O(\sqrt{\log N})$.  This approximation scheme involves a different 
disjunctive formulation that can be rounded to obtain the desired approximation.  Specifically, we 
construct a disjunction over $O(N((1/\epsilon + T)^{O(\log(T/\epsilon)/\epsilon^2)}))$ LPs, each with $NT$ variables and $O(NT)$ constraints. This improves on the result of Sharp \cite{Sharp07}.  Moreover, the analysis of the approach extends for certain $\Delta_t$ such as when $\Delta_t = e^{-rt}$ for some $r > 0$.  This allows us to incorporate discounting.
This PTAS also extends the earlier work of Bienstock and McClosky  \cite{Bienstock08}, \cite{BienstockM12} on the disjunctive approach for the single period knapsack problem.

Both of our algorithms rely on the classical approach of disjunctive programming \cite{Balas75}. Suppose we want to 
find  an approximate solution to $\max \{ w^{T}x \, : \, x \in P\}$ ($P \subseteq \R^n$, possibly non-convex), with approximation factor $k$. Moreover, the difficulty of the problem lies in that no good convex relaxation of $P$ is known. In this case, we may still be able to leverage the idea of disjunctive programming to give us a good approximation guarantee. The idea is to find a set of polyhedra $Q^{1}, Q^{2}, \ldots Q^{L}$ in $\R^n$ such that $P \subseteq \cup_{i=1}^{L} Q^{i}$ and for each $i$ we can compute, in polynomial time, $x^{i} \in P$ with $w^{T}x^{i} \geq k \max \{ w^{T}x \, : \, x \in Q^{i}\}$. Taking $x^{\star} = \argmax_{i} \{ w^{T}x^{i} \}$ yields a factor $k$ approximate solution to the original optimization problem. As stated, this approach simply constitutes a case of enumeration (polynomially-bounded if $L$ is polynomial). Further, $w^T x^{\star} \ge k \max \{ w^T x \, : \, x \in \conv( \cup_i Q^i ) \}$,  and this last maximization problem can be formulated as a single linear program (polynomial-sized if $L$ is), and, as will be the case below, we obtain an an approximation algorithm based on rounding.

The rest of the paper is organized as follows. In section \ref{LP relax}, we show that the natural IP formulation of the time-invariant incremental knapsack problem has an unbounded integrality gap. In section \ref{constant alg}, we give the constant factor approximation algorithm for $\IK$. In section \ref{PTAS}, we show a PTAS for $\IIK$. In section \ref{Conclusion}, we summarize our results and give suggestions for future directions.

\section{LP Relaxation and Integrality Gap} \label{LP relax}
To motivate the disjunctive approach, we will first show that the LP relaxation of a natural IP formulation for $\IIK$  has an unbounded integrality gap. This result implies that any constant factor approximation algorithm must do something more clever than simply solving the LP relaxation and rounding the fractional solutions to a feasible integral solution.

Let $x_{i,t} = 1$ if item $i$ is placed in the knapsack at time $t$ and $0$ otherwise. We can formulate $\IK$ as the following binary integer program, whose feasible region will be denoted by $P$. 

\begin{align}
\text{IP = max} \quad \label{IKP IP}
\sum_{t =1}^{T} \sum_{i=1}^{N} v_i  \Delta_t x_{i,t}\\
\text{s.t.} \quad
\sum_{i = 1}^{N} w_ix_{i,t} \leq B_t \quad &\forall t \notag \\
x_{i, t-1} \leq x_{i, t} \quad \forall i, \text{and} &\ t = 2, 3, \ldots, T \notag \\
x_{i, t-1} \in \{0,1\} \quad \forall i, t. \notag \\  \notag
\end{align}

Consider the case where $B_t = t$ for all $t$,
$N = 1$ and $v_1 = w_1 =T$. Further assume $\Delta_t = 1$ for all $t$, i.e. we are consider the time-invariant case. Clearly (\ref{IKP IP}) has value $T$ in this instance, whereas the LP relaxation gives a value of $\frac{T(T+1)}{2}$, which implies that the integrality gap of $O(T)$ which is unbounded as $T \rightarrow \infty$. 

A natural idea is to strengthen the LP relaxation by setting $x_{it} = 0$ if item $i$ does not fit into the knapsack at time $t$. This strengthened LP relaxation, shown below, still has an unbounded integrality gap (as shown in the Appendix).

\begin{theorem} \label{Integrality Gap}
The following LP relaxation to $\IIK$  has an integrality gap that cannot be bounded by any constant. 
\begin{align}
\text{max} \quad
\sum_{t =1}^{T} \sum_{i=1}^{N} v_ix_{i,t} \notag \\
\text{s.t.} \quad
\sum_{i = 1}^{N} w_ix_{i,t} \leq B_t \quad &\forall t \notag \\
x_{i, t-1} \leq x_{i, t} \quad \forall i, \text{and} &\ t = 2, 3, \ldots, T \notag \\
x_{i,t} = 0 \ \quad \text{for any $i$,} & \ \text{$t$ such that } \ w_i > B_t \notag \\
x_{i, t-1} \in [0,1] \quad \forall i, t. \notag \\  \notag
\end{align}
\end{theorem}

\section{Constant Factor Algorithm} \label{constant alg}
In this section we provide a constant-factor approximation algorithm for $\IK$ when the capacity 
function $B_t$ is upper bounded by a polynomial function of $t$.  To motivate our approach we introduce two definitions.

\begin{definition}  \label{split} Let $S > 0$ and $0 < \kappa \le 1$.  We say that a $(S,\kappa)$-split takes place
at time $t_{\kappa}$ if the following conditions happen:
\begin{itemize}
\item [(i)] $\sum_{t = 1}^{t_{\kappa}-1} \Delta_t \le S \sum_{t = t_{\kappa}}^T \Delta_t$.
\item [(ii)] $B_{T} - B_{t_{\kappa}} \le \kappa \, B_{T}$.
\end{itemize}
\end{definition}
\noindent {\bf Remark 1.} Consider the time-invariant case, i.e. $\Delta_t = 1$ for all $t$.  When an $(S,\kappa)$-split 
takes place at $t_{\kappa}$, we have $t_{\kappa} < S T$, but the knapsack capacity 
at time $t_{\kappa}$ is already (at least) a  fraction $1 - \kappa$ of the final knapsack capacity, $B_T$.

\begin{definition}  \label{split2} Let $0 < \kappa \le 1$.  Define $S(\kappa)$ to be the smallest value $S \ge 0$, such
that there is a period $t_{\kappa}$ where an $(S, \kappa)$ split takes place. \end{definition}

\noindent {\bf Remark 2.} Note that $S(\kappa)$ decreases with $\kappa$.  Also, suppose $\Delta_t = 1$ for all $t$  and $B_t = \Theta(t^p)$ for some $p > 0$.  Then we have that $S(\kappa) \approx \frac{(1 - \kappa)^{1/p}}{1 - (1 - \kappa)^{1/p}}$.  Even though
this quantity converges to $ +\infty$ as $p \rightarrow + \infty$, for \textit{fixed} $p$ it is bounded.  This means that
(by Definition \ref{split}) during the last $\frac{T}{S(\kappa)}$ time periods the knapsack only gains a fraction $\le \kappa$ of the
final capacity $B_T$.\\

\noindent Given $\kappa > 0$, there is an algorithm with approximation factor $\min\{ \Omega(\frac{ 1 - \kappa }{\kappa}), \frac{1}{S(\kappa)} \}$. Thus, roughly speaking, the quality of the approximation improves if both $\kappa$ and $S(\kappa)$ are ``small''. Or,
to put it differently, if the capacity function $B_t$ is such that $S(\kappa)$ remains very large for $\kappa \approx 1$,
then the quality of the approximation bound will suffer.  As remark 2 shows, when $B_t$ is polynomially bounded
as a function of $t$, $S(\kappa)$ is bounded above, and so the ratio  $\frac{ 1 - \kappa }{\kappa S(\kappa)}$ remains bounded away from $0$, i.e. we indeed
obtain a constant-factor approximation algorithm.    The class of cases where we do not attain a constant-factor
approximation algorithms are those where $\lim_{\kappa \rightarrow 1}  S(\kappa) = + \infty$.  In such cases, the 
capacity function is attaining essentially all of its growth in an arbitrarily small final set of time periods.\\

In this paper,  for brevity,  we focus on the case $\kappa = 1/2$, and prove the following.
\begin{theorem} \label{halflemma} There is a polynomial-time algorithm for $\IK$ with approximation ratio $\min\left\{ \frac{1}{9}, \frac{1}{6 \max\{1, S(1/2)\}} \right\}$.\end{theorem}  
In the Appendix we outline how to
extend our approach to general $0 < \kappa < 1$. \\

\noindent {\bf Simplification.} For simplicity, we assume, by perturbing the $v_i$ if necessary that $v_i/w_i \neq v_j/w_j$ for all $i \neq j$.
The perturbation changes the value of $\IK$ by an arbitrarily small amount.
\subsection{The approximation algorithm}
Our algorithm is based on running (up to) two polynomially bounded procedures given below, for each time period $1 \le \bar t \le T$: a ``replicated  knapsack'' procedure, and an LP-rounding algorithm. Each run
will produce a feasible (integral) solution to $\IK$; the algorithm will select a solution that attains the highest objective
value.    We will use the same notation as above, i.e. the variable $x_{i,t}$ is used
to indicate whether item $i$ is in the knapsack at time $t$.

To introduce the first procedure we need a definition.
\begin{definition} \label{replicated} Let $1 \le \bar t \le T$.  A 0-1 vector $\bar x$ is a {\bf replicated-knapsack solution} at $\bar t$ for $\IK$ if the following conditions hold:
\begin{eqnarray}
&& \bar x_{i,t} = 0 \  \mbox{for all $i$ and $t < \bar t$}, \quad \mbox{and} \quad \bar x_{i,t} = \bar x_{i, \bar t} \ \mbox{for all $i$ and $t  \ge \bar t$.} \nonumber
\end{eqnarray}
\end{definition}
Recall that the capacity function $B_t$ is monotonely nondecreasing.  Thus the replicated-knapsack solution at $\bar t$ is
feasible for $\IK$ iff $\sum_i w_i x_{i, \bar t} \le B_{\bar t}$, i.e. the solution replicates a feasible solution to the (single-period) knapsack with weights $w_i$ and capacity $B_{\bar t}$.  Further, the objective value of $\bar x$ 
equals $ \left( \sum_{t = \bar t}^T \Delta_t \right) \sum_{i=1}^{N} v_i \bar x_{i, \bar t}$.
This quantity can be approximated (arbitrarily closely) in polynomial time using several well-known methods.

In order to describe the second procedure, consider formulation (\ref{vobj})-(\ref{bottpart}), for which the time period $1 < \bar t < T$,  
the item index $h$, and a second time period $\breve t \le \bar t$ are inputs. Here, $h$ is the most valuable item in the knapsack at time $\bar t$ and $h$ was first placed in the knapsack at time $\breve t$.
\begin{eqnarray}
\mbox{\bmath{D(\bar t, \breve t, h):}} \quad && \quad \max \sum_{t =1}^{T} \sum_{i=1}^{N} v_i \Delta_t x_{i,t} \label{vobj}\\
\text{s.t.} && \quad \sum_{i = 1}^{N} w_ix_{i,t} \leq B_t \quad \forall t \label{capacity} \\
&& x_{i, t-1} \leq x_{i, t} \quad \forall i, \ \text{and} \ t = 2, 3, \ldots, T \label{monotonoex} \\
&& x_{i, t} \in \{0,1\} \quad \forall i, t \label{01constr} \\  
&& x_{h, \breve t} \ = \ 1, \quad x_{h, \breve t - 1} \ = \ 0 \ \mbox{(if $\breve t > 0$)} \quad \mbox{and} \notag \label{enumhighesta}\\
&& x_{i, \bar t} \ = \ 0, \quad \mbox{$\forall$ $i$ with $v_i > v_h$}  \label{enumhighestb}\\
&& \frac{1}{3} \sum_{i, t} v_i \Delta_t x_{i,t} \ \le \ \sum_{t = \bar t}^T \sum_{i=1}^{N} v_i \Delta_t x_{i,t} \label{toppart}\\
&& \frac{1}{3} \sum_{i, t} v_i \Delta_t x_{i,t} \ \le \ \sum_{t = 1}^{\bar t - 1} \sum_{i=1}^{N} v_i \Delta_t x_{i,t}. \label{bottpart}
\end{eqnarray}
Other than constraints (\ref{enumhighesta})-(\ref{bottpart}) this is the same formulation as in Section 
\ref{LP relax}.  Constraints (\ref{toppart}) and (\ref{bottpart}) indicate that, in the time interval $[\bar t, T]$
(resp. $[1, \bar t -1]$) at least one-third of the total value is accrued.  Using (\ref{monotonoex}) and $\breve t \le \bar t$, constraint (\ref{enumhighesta}) implies that $x_{h, \bar t} = 1$.  Thus,
(\ref{enumhighestb}) states that
item $h$ is placed in the knapsack at time $\breve t$, and that 
at time $\bar t$, item $h$ is a highest-value item in the knapsack.     
We will now prove a number of results on this
formulation.  

\begin{definition} Let $D^R(\bar t, \breve t, h)$ denote the continuous relaxation of 
$D(\bar t, \breve t, h)$, i.e. that obtained by replacing (\ref{01constr}) with $0 \le x_{i,t} \le 1$ for all $i, t$.
Further, let $V^* = V^*(\bar t, \breve t, h)$ denote the value of $D^R(\bar t, \breve t, h)$.
\end{definition}

\begin{lemma}\label{onefrac}  Suppose $\tilde x \, = \, \tilde x(\bar t, \breve t, h) \, \in \R^{N \times T} $ be an optimal solution to the linear program where:
\begin{itemize}
\item [(i)] The objective function is (\ref{vobj}).
\item [(ii)] The constraints are (\ref{capacity})-(\ref{enumhighestb}), together with
\begin{eqnarray}
&& \frac{1}{3} V^* \ \le \ \sum_{t = \bar t}^T \sum_{i=1}^{N} v_i \Delta_t \tilde x_{i,t} \quad \mbox{and} \quad \frac{1}{3} V^* \ \le \ \sum_{t = 1}^{\bar t - 1} \sum_{i=1}^{N} v_i \Delta_t \tilde x_{i,t}. \label{top}
\end{eqnarray}
\end{itemize}
i.e. we replace constraints (\ref{toppart}) and (\ref{bottpart}) of $D^R(\bar t, \breve t, h)$ with (\ref{top}). Then for any time period $t$ there
is at most one item $i$ with $0 < \tilde x_{i,t} < 1$. \end{lemma}
\noindent {\bf Remark.} It can be shown that $\tilde x$ can be obtained by solving a single, polynomial-size linear program, rather than the two-LP procedure implied by Lemma \ref{onefrac}.

\begin{lemma}\label{rounddown}  Let $\tilde x = \tilde x(\bar t, \breve t, h)$ be as in Lemma \ref{onefrac} and let $\breve x$ be obtained by
rounding down $\tilde x$.  Then $\breve x$ is feasible for $D(\bar t, \breve t, h)$, and
\begin{eqnarray}
&& \sum_{i,t} v_i \Delta t \, \breve x_{i,t} \ \ge \ \frac{1}{6} \min  \left\{1, \frac{\sum_{t = \bar t}^{T} \Delta_t}{ \sum_{t = 1}^{\bar t -1} \Delta_t} \right\}  \sum_{i,t} v_i \Delta t \, \tilde x_{i,t}.  \label{niceroundown}
\end{eqnarray}
\end{lemma}

\hspace{.1in} 

Lemma \ref{rounddown}, together with the replicated-knapsack construction given above, constitute the two cases that
our algorithm will enumerate for each $t$.
We will next show that the best solution of the replicated-knapsacks and 
the roundown constructions attain a factor $\frac{1}{6 S}$ where $S = \max\{1, S(1/2) \}$ (recall Definition \ref{split}).

\subsection{Existence of approximation}

In the remainder of this section, we assume that $x^Z$ is an {\bf optimal} solution to a given instance of
$\IK$.  We $V^Z$ denote the value of the instance, i.e. $ V^Z \doteq \sum_i v_i \Delta_t x^Z_{it}$.

\begin{lemma} \label{onethird} Either there is a period $1 < t^{[3]} < T$ such that 
\begin{itemize}
\item [(a)] $\frac{1}{3}V^Z \ \le \ \sum_{t = t^{[3]}}^T \sum_i v_i \Delta_t x^Z_{it}$
\item [(b)] $\frac{1}{3}V^Z \ \le \ \sum_{t = 1}^{t^{[3]} - 1} \sum_i v_i \Delta_t x^Z_{it}$
\end{itemize}
or there is a replicated-knapsack solution with value at least $ \frac{V^Z}{3}$.
\end{lemma}

Using Lemma \ref{onethird} we can assume that none of the $T$ replicated-knapsack solutions (which we can approximate,
in polynomial time, to any desirable constant factor) is within a factor of $1/3$ of $V^Z$, and thus, that a time
period $1 < t^{[3]} < T$ satisfying (a) and (b) does exist.  Using this fact, we will next
consider the formulations 
$D(\bar t, \breve t, h)$ discussed in the previous section, and prove that the bound obtained in Lemma \ref{rounddown} will
yield a large enough constant factor for at least one choice of $\bar t$ and $h$.

Recall Definition \ref{split}. Let $S = S(1/2)$ and $t_{1/2}$ be a time period so that
\begin{eqnarray}
&& \sum_{t = 1}^{t_{1/2}-1} \Delta_t \le S \sum_{t = t_{1/2}}^T \Delta_t \quad \mbox{and} \quad B_{T} - B_{t_{1/2}} \le \frac{B_{T}}{2}.
\label{Sagain}
\end{eqnarray}
\begin{lemma} \label{hardcase} Suppose that $t^{[3]} \le t_{1/2}$.  Let $h^*$ be the most valuable item in
the knapsack, under solution 
$x^Z$, at time $t^{[3]}$, and let $\breve t \le t^{[3]}$ be the time it was placed in the knapsack.  Then
rounding down the solution to $D^R(t^{[3]}, \breve t, h^*)$ yields a feasible solution to problem $\IK$, of value at least
$$\frac{1}{6 \max\{1, S\}} \sum_{i,t} v_i \Delta t \, x^Z_{i,t}. $$
\end{lemma}
\noindent {\em Proof.} By Lemma \ref{rounddown} applied to formulation $D(t^{[3]}, \breve t, h^*)$, 
we will obtain a feasible
solution to $\IK$ with value at least
\begin{eqnarray}
&& \frac{1}{6} \min \left\{1, \frac{\sum_{t = t^{[3]}}^{T} \Delta_t}{ \sum_{t = 1}^{t^{[3]} -1} \Delta_t} \right\} \sum_{i,t} v_i \Delta t \, x^Z_{i,t}. \label{finally}
\end{eqnarray}
But we are assuming that $t^{[3]} \le t_{1/2}$.  This implies the desired bound. \qed

\hspace{.1in} 

We can now assume that $t_{1/2} < t^{[3]}$.  We can show that in this case a replicated-knapsack
solution has value at least $V^*/9$. 

\begin{lemma} \label{easycase} Suppose $t_{1/2} < t^{[3]}$.  Then a replicated-knapsack
solution has value at least $V^*/9$.
\end{lemma}

\section{A PTAS for $\IIK$ when $T = O(\sqrt{\log N})$} \label{PTAS} 
Now we are ready to present the PTAS for the time-invariant incremental knapsack problem when $T = O(\sqrt{\log N})$. This algorithm is easily extended to the case of $\IK$ with fixed $T$ and monotonically nonincreasing $\Delta_t$ quantities.  Consider an instance of $\IIK$ and let $ \epsilon \in (0,1)$. Without loss of generality, we can assume that the $v_i$'s are integral. Fixing an optimal solution $OPT$, and let $h$ be a the maximum valued item that is ever placed in the knapsack by $OPT$. Then it suffices to optimize over the set of items $S^{h} = \{ i \in S | v_i \leq v_h\}$.  We partition $S^{h}$ into $K+1$ subsets $X = \{S^{1,h}, S^{2,h}, \ldots, S^{K,h}, T^{h}\}$, where 
\[S^{k,h} = \{j \in S, j \neq h: \ (1-\epsilon)^{k-1}v_h \geq  v_j > (1-\epsilon)^{k}v_h\} \quad \ \text{for} \ k = 1, \ldots, K,\]
and
\[T^{h} = \{j \in S : (1- \epsilon)^Kv_h \geq v_j\}.\]
In order to attain the approximation ratio, we will choose $K$ large enough so that $(1-\epsilon)^{K} < \epsilon /T$ or equivalently, $K > \frac{\log (T/\epsilon)}{\epsilon}$. 

Consider a modified instance of the problem where items have identical weights as the original instance and item $i$ has a modified value of $v'_i = (1-\epsilon)^{k-1}V$ if $i \in S^k$ and $v'_i = v_i$ otherwise. Let $OPT_m$ denote an optimal solution to the modified instance of the problem. Let $V(SOL)$ and $V_m(SOL)$ be the objective value with respect to a solution $SOL$ of the original instance and the modified instance respectively. As we did not change the item weights, $OPT_{m}$ is a feasible solution to the original instance. Moreover, 
\[V(OPT_{m}) \geq (1-\epsilon)V_m(OPT_m) \geq (1-\epsilon)V_m(OPT) \geq (1-\epsilon)V(OPT),\] 
where the first inequality follows from the fact that $v_i \geq (1-\epsilon)v'_i$ for every item $i$, the second inequality follows from the fact that $OPT_m$ is an optimal solution to the modified instance, and the third inequality follows from the fact that $v_i \leq v'_i$ for every item $i$. 

Now, since all items within each $S^{k,h}$ have equal value in the modified instance, it is clear that conditioning on the number of items chosen by OPT within each $S^{k,h}$, OPT would tend to choose the items within the same value class in the order of non-decreasing weight (breaking ties arbitrarily). Thus, it suffices to enumerate feasible solutions that can be described by a collection of $K$ $T$-vectors $\{\sigma^{1}, \ldots, \sigma^{K}\}$, where $\sigma_{t}^{k} \in \{0, 1, \ldots, |S^{k, h}|\}$ denotes the number of items chosen from $S^{k,h}$ in time period $t$, in order to find an optimal solution. Nonetheless, the number of potential solutions that we have to enumerate would be exponential in $N$ if we attempt to enumerate  all possible configurations of $\{\sigma^{1}, \ldots, \sigma^{K}\}$. Consequently, we will only explicitly enumerate $\sigma_{t}^{k}$ taking values from $\{0,1,\ldots, \ \min(\lceil 1/\epsilon \rceil, |S^{k,h}|)\}$. For $\sigma_{t}^{k}$ taking values larger than $J = \lceil 1/\epsilon \rceil$, we will instead let the feasible region of an LP capture these feasible points and try to let the LP choose an optimal solution for us and subsequently round this optimal solution. Lastly, since we don't know the most valuable item $h$ taken by $OPT$ in the original instance of the problem, we will have to guess such an item by enumeration.

Our disjunctive procedure is as follows. First, we guess the most valuable item $h \in S$ packed by an optimal solution. Subsequently, we only consider choosing items from $S^{h}$ and round the values of the items in $S^{h}$ to obtain the modified instance of the problem. We will then focus on solving the modified instance of the problem.  Let $k_i$, $i = 1, 2, \ldots, |S^{k,h}|$, be the $i$-th lightest weight item in $S^{k,h}$ (break ties arbitrarily). Let $x_{k_i,t}$ be the variable indicating whether item $k_i$ is placed in the knapsack in time period $t$.  Let $\sigma = \{\sigma^{1}, \ldots, \sigma^{K}\} \in \{0,\ldots, J\}^{TK}$ and define the following polyhedron:

\[
\begin{split}
Q^{\sigma, h}= \{&x \in [0,1]^{T|S^{h}|}: x_{h,T} = 1 \\
& x_{k_1,t} = x_{k_2,t} = \ldots = x_{k_{|S^{k,h}|},t} = 0 \ \quad \forall (k,t) \ \text{s.t.} \ \sigma^{k}_{t} = 0 \\
& x_{k_1,t} = x_{k_2,t} = \ldots = x_{k_{\sigma^{k}_{t}},t} = 1, \ x_{k_{\sigma^{k}_{t} + 1},t} = \ldots = x_{k_{|S^{k,h}|},t} = 0\  \\ &\quad \forall (k,t) \ \text{s.t.} \ 1 \leq \sigma^{k}_{t} < J \ \text{and} \ \sigma^{k}_t < |S^{k,h}| \\
& x_{k_1,t} = x_{k_2,t} = \ldots = x_{k_{\sigma^{k}_{t}},t} = 1, \ \sum_{i=1}^{|S^{k,h}|} x_{k_i, t} \geq \sigma^{k}_{t} \\ &\quad \quad \quad \quad \quad \quad \quad \quad \forall (k,t) \ \text{s.t.} \ \sigma^{k}_{t} = J \ \text{and} \ \sigma^{k}_t < |S^{k,h}| \\
& x_{k_1,t} = x_{k_2,t} = \ldots =  x_{k_{|S^{k,h}|},t} = 1\  \quad \forall (k,t) \ \text{s.t.} \ \sigma^{k}_t \geq |S^{k,h}| \\
& \sum_{k=1}^{K}\sum_{i = 1}^{|S^{k,h}|} w_{k_i}x_{k_i,t} + \sum_{i \in T^{h}} w_i x_{i,t} \leq B_t \quad \forall t \\
& x_{k_i, t-1} \leq x_{k_i, t} \quad \forall (k, i), \text{and} \ t = 2, 3, \ldots, T \\
& x_{i, t-1} \leq x_{i, t} \quad \forall i \in T^{h}, \text{and} \ t = 2, 3, \ldots, T \}. \\ 
\end{split}
\]
\begin{lemma} \label{NumDisjunctions}
For each fixed $h$, there are $O((1/\epsilon + T)^{O(\log(T/\epsilon )/\epsilon^2)})$ polyhedra $Q^{\sigma,h}$ in our disjunctive procedure.
\end{lemma}
Since we have to enumerate our guess for the most valuable item $h$, we get the following corollary.
\begin{corollary}
There are a total of $O(N(1/\epsilon + T)^{O(\log(T/\epsilon )/\epsilon^2)})$ LPs in our disjunctive procedure.
\end{corollary}
Notice that when $T = O(\sqrt{\log T})$, then the number of LPs is polynomial in $N$ for a fixed $\epsilon$. \\

The following is our main result.

\begin{theorem} \label{Disjunctive procedure theorem}
For every non-empty polyhedron $Q^{\sigma, h}$,  there exists a polynomially computable point $x^{\sigma,h}$ feasible for $\IIK$,  such that  \[\sum_{t=1}^{T}\sum_{i \in S^{h}} v'_{i}x^{\sigma,h}_{i,t} \geq (1- \epsilon) \max \{ \sum_{t=1}^{T}\sum_{i \in S^{h}} v'_{i}x_{i,t}: \ x \in Q^{\sigma,h}\}.\] 
\end{theorem}

\begin{theorem}
Let $x^{\star} = \arg \max_{\ Q^{\sigma,h} \neq \emptyset} \sum_{t=1}^{T}\sum_{i \in S^{h}} v'_{i}x^{\sigma,h}_{i,t}$, where $x^{\sigma,h}$ is defined in the previous theorem, then we have that 
\[\sum_{t=1}^{T}\sum_{i \in S^{h}} v'_{i}x^{\star}_{i,t} \geq (1-\epsilon) V_m(OPT_m) \geq (1-\epsilon)^2 V(OPT).\]
\end{theorem}
\begin{proof}
This is a direct consequence of the fact that $OPT_m \in \cup_{Q^{\sigma,h} \neq \emptyset} Q^{\sigma,h}$ as $\cup_{Q^{\sigma,h} \neq \emptyset}Q^{\sigma,h}$ covers $P$,  and that 
\[\sum_{t=1}^{T}\sum_{i \in S^{h}} v'_{i}x^{\star}_{i,t} \geq (1-\epsilon)\max \{\sum_{t=1}^{T}\sum_{i \in S^{h}} v'_{i}x_{i,t}: \ x \in \cup_{Q^{\sigma,h} \neq \emptyset}Q^{\sigma,h}\}.\]
\end{proof}

\section{Conclusion} \label{Conclusion}
In this work, we give a constant factor approximation algorithm for $\IK$ when the capacity function $B_t$ is upper bounded by a polynomial function of $t$. We also give a PTAS for $\IIK$ when the time horizon $T = O(\sqrt{\log N})$, where $N$ is the number of items. Our results generalize and improve on some of the earlier results of Hartline and Sharp for this problem. Our work leaves to the following open questions. First, is there a polynomial time algorithm for $\IK$ with a constant factor approximation ratio that makes no assumption on the growth rate of $B_t$? Second, is there a PTAS for $\IIK$ for an arbitrary time horizon $T$? It is also interesting to consider an incremental version of other combinatorial optimization problems.


\bibliographystyle{plain}
\bibliography{paper}

\newpage
\normalsize

\section{\appendixname}

\subsection*{Proof of Proposition~\ref{Strongly NP-hard}}

\begin{proof}
In the $3$ partition problem, we're given a set $S$ of $3m$ integers $a_1, \ldots a_{3m}$, and we want to decide whether $S$ can be partitioned into $m$ triples that all have the same sum $B$, where $B = (1/m)\sum_{i=1}^{3m} a_i$.  It has been shown that $3$-partition is strongly NP-hard even if all the integers are between $B/4$ and $B/2$. We will reduce any instance of the 3-partition problem to a corresponding instance of the incremental subset problem. Given a set $S$ of $3m$ integers $a_1, \ldots, a_{3m}$. Let $a_i$ be the weight/value of item $i$. And we have a knapsack whose capacity is $B_t$ in period $t$ for $t = 1, \ldots, m$. Notice that the partition can be done if and only if in every period the knapsack reaches its capacity, i.e. the incremental amount that we pack is $B$. Moreover, it can be shown inductively that since the value of the items are strictly in between $B/4$ and $B/2$, we would pack exactly $3$ additional items in every period.
\end{proof}

\subsection*{Proof of Theorem~\ref{Integrality Gap}}
\begin{proof}
Fix a $k \geq 2$ and let $T$ be a power of $k$. Consider a set of $N$ items, where $v_i = w_i = k^i$ for $i = 1,\ldots, \log_k(T) = N$. The knapsack capacities follow the following pattern: \[B_t = k^{i} \ \text{if} \ T(1-\frac{1}{k^{i-1}})+1\leq t \leq T (1- \frac{1}{k^{i}}) \ \text{for} \ i = 1,\ldots, \log_k(T) \ \text{and} \ B_{T} = B_{T-1}.\]  Since the LP can fractional pack the items, the knapsack attains its capacity in every time period. Moreover, since all items have weight equaling value, we have that the optimal value of the LP solution is (by evaluating the sum of the knapsack values over all time periods):
\[T + T \sum_{i=1}^{\log_k(T)}k^iT((1-1/k^{i}) - (1-1/k^{i-1}) = T(k-1)\log_k(T) + T = O(Tk\log_k (T)).\]
Let $t_i = T(1-1/k^{i-1})+1$ denote the first time when the knapsack capacity increases to $k^{i}$. Notice that any integer optimal solution would only pack additional items in time those periods, which means that we just need to figure out what to pack at those time periods in order to find an integer optimal solution. The only items that fit in the knapsack at time $t_i$ are items $0$ through $i$. If we decide to pack item $i$ in period $t_{i}$, then the total revenue we get for packing $i$ over times $t_{i} \leq t \leq t_{i+1} - 1$ is $T(1/k^i - 1/k^{i+1})k^{i+1} =T( k-1)$. Since we cannot pack any item before time $t_{i}$, the total revenue we get up to time $t_{i+1} - 1$  if we pack item $i$ in time $t_{i}$ would be $T(k-1)$. For every $i > 1$, this is clearly suboptimal since we would get more revenue up to time $t_{i+1}-1$ had we just packed item $1$ in period $1$ (since $ kT(1-1/k^{i+1}) > kT(1-1/k) = T(k-1)$ for $i > 1$). Hence, no integer optimal solution would pack item $i$ at time $t_i$ for every $i > 1$.\\ 

If we do not pack item $i$ at time $t_{i}$, then we can pack the first $i-1$ items at time $t_{i}$ for every $i > 1$ since $k\geq2$. Hence, this is an optimal packing for time $t_i$ for every $i > 1$ (and it is optimal to pack item $1$ starting from period $1$). Moreover, this solution respects the precedence constraints, which means that we have found an integer optimal solution. We evaluate the integer optimal solution by looking at how long each item has been placed in the knapsack:
\[
\begin{split}
kT + \sum_{i = 2}^{N-1} k^i(T - t_{i+1} + 1) &= kT + \sum_{i = 2}^{N-1} k^i(T - T(1-1/k^{i})) \simeq kT + T (\log_k T-1) \\
&= O(T\max(k, \log_k (T))).
\end{split}
\] 

Hence, the integrality gap is at least $\min(\log_k(T), k)$. For every $k$, we can choose $T$ appropriately so that $\min(\log_k(T), k) = k$.  Letting $k$ go to infinity and we have the desired result.  
\end{proof}

\subsection*{Proof of of Lemma \ref{onefrac}}
\begin{proof}
Let $\hat t$ be a time period such that there exist items $i$, $j$ with $\tilde x_{i,\hat t}$
and $\tilde x_{j,\hat t}$ both fractional.  Without loss of generality assume that
$v_i/w_i > v_j/w_j$.  Let $t^1 = \min\{ t \, : \, \tilde x_{j,t} > 0 \}$, and 
$\epsilon = \tilde x_{j, t^1}$.
Let 
$t^2 = \max\{ t \, : \, \tilde x_{i, t} = \tilde x_{i, \hat t}\}$,   
and when $t^2 < T$ set $\delta = \tilde x_{i, t^2 + 1} - \tilde x_{i, \hat t}$ and otherwise set $\delta = 1 - \tilde x_{i, \hat t}$.  Finally, write $\theta = \min\left\{\frac{w_i}{w_j}\delta, \epsilon\right\}$. Consider
the vector $y$ created by the following rule:
\begin{eqnarray}
&& y_{j, t} = \tilde x_{j, t} - \theta, \quad \ \, \mbox {for $t^1 \le t \le t^2$}, \nonumber \\
&& y_{i, t} = \tilde x_{i, t} + \frac{w_j}{w_i}\theta,  \ \mbox {for $t^1 \le t \le t^2$}, \nonumber \\
&& y_{k, t} = \tilde x_{k, t},  \quad \quad \quad \mbox {for all remaining $k$ and $t$}. \nonumber 
\end{eqnarray}
Note that $i, j \neq h$ because constraint (\ref{enumhighesta}) of the formulation guarantees that $\tilde x_{h, t}$ is integral for all $t$.  Thus, $y$ is a feasible solution to $D^R(\bar t, \breve t, h)$. But since $v_i/w_i > v_j/w_j$ 
(because we have perturbed data so that all values/weight ratios are distinct)
the objective value attained by $y$ is strictly larger than that of $\tilde x$, a contradiction. 
\end{proof}

\subsection*{Proof of of Lemma \ref{rounddown}}
\begin{proof}  For any period $t$, let $F(t) \doteq \{ i \, : \, 0 < \tilde x_{i,t} < 1 \}$.  
Suppose first 
that
$$ \sum _{t = 1}^{\bar t - 1} \sum_{i \in F(t)} v_i \Delta_t \tilde x_{it} < \frac{1}{6} V^*.$$
By constraint (\ref{bottpart}) we therefore have
\begin{eqnarray}
&& \sum _{t = 1}^{\bar t - 1} \sum_{i \notin F(t)} v_i \Delta_t \breve x_{it} > \frac{1}{6} V^* \notag
\end{eqnarray}
and 
\begin{eqnarray}
&& \sum_{i,t} v_i \Delta t \, \breve x_{i,t} \ \ge \ \min \frac{1}{6} \left\{1, \frac{\sum_{t = \bar t}^{T} \Delta_t}{ \sum_{t = 1}^{\bar t -1} \Delta_t} \right\}  \sum_{i,t} v_i \Delta t \, \tilde x_{i,t}.  \label{niceroundown2}
\end{eqnarray}
(i.e. (\ref{niceroundown})) follows, which is the desired result. Thus, we instead assume that
$$ \sum _{t = 1}^{\bar t - 1} \sum_{i \in F(t)} v_i \Delta_t \tilde x_{it} \ge \frac{1}{6} V^*.$$
In particular this means that $F(t) \neq \emptyset$ for at least one $1 \le t \le \bar t -1$.  But, for such 
$t$, if $i \in F(t)$ then $\tilde x_{i,\bar t} > 0$ and by (\ref{enumhighesta}) and (\ref{enumhighestb}), 
$$ v_i \le v_{h}.$$
So, using Lemma \ref{onefrac},
\begin{eqnarray}
\frac{1}{6} V^* & < & \sum _{t = 1}^{\bar t - 1} \sum_{i \in F(t)} v_i \Delta_t \tilde x_{it} \ \le \ \sum _{t = 1}^{\bar t - 1} v_{h} \Delta_t \ \le \ \frac{ \sum_{t = 1}^{\bar t -1} \Delta_t}{\sum_{t = \bar t}^{T} \Delta_t} \sum_{t = \bar t}^T \sum_i v_i \Delta t\, \breve x_{i, t}.
\end{eqnarray}
From this relationship (\ref{niceroundown2}) follows. 
\end{proof}

\subsection*{Proof of of Lemma \ref{onethird}}
\begin{proof}
Define $t^{[3]}$ to be the largest period so that (a) holds.  If $t^{[3]} = T$ then the
replicated-knapsack solution at $T$ has value at least $ \frac{V^Z}{3}$.  So we can assume $t^{[3]} < T$ (and similarly,
that $1 < t^{[3]})$ and therefore
$$ \sum_{t = 1}^{t^{[3]} - 1} \sum_i v_i \Delta_t x^Z_{it} < \frac{1}{3}V^Z.$$
Hence if (b) does not hold, then
$$ \sum_i v_i \Delta_t x^Z_{it^{[3]}} > \frac{1}{3}V^Z,$$
and so the replicated-knapsack solution at $t^{[3]}$ has value at least $ \frac{V^Z}{3}$. 
\end{proof}

\subsection*{Proof of of Lemma \ref{easycase}}
\begin{proof}
We observe that if, under $x^Z$, some item $i$ is added to the knapsack 
at time $t$, then, since $x^Z$ is optimal and $v_i > 0$, it must be the case that it could not have been added any earlier (while keeping
the remaining schedule fixed, otherwise).  Under this assumption, the schedule is ``pushed to the
left'' as much as possible.

Suppose first that no items are added to the knapsack, under $x^Z$, in the periods $t^{[3]}, \ldots, T$. 
Then by condition (a) of Lemma \ref{onethird} we have that by replicating the
knapsack solution at time $t^{[3]}$ we will have total value is at least $V^*/3$, as desired.
 
Thus, let $t_1$ be the first period $\ge t^{[3]}$ such that $x^Z$ adds an item is added to the knapsack
at that period. By the ``pushed to the left'' analysis, we have that 
either $t_1 = t^{[3]}$, or
\begin{eqnarray}
&& t_1 > t^{[3]}, \quad \mbox{and} \nonumber \\
&& \sum_i w_i x_{i, t_1} > B_{t^{[3]}}. \label{high1}
\end{eqnarray}

Let $t_2$ be the first time period after $t_1$ where $x^Z$ adds an item knapsack.
If no such period exists, let $t_2 = T$. Finally let $A(1)$ be the set of items added by $x^Z$ to the knapsack 
in period $t_1$, and $A(2)$ the set of items added in periods $t_2$ through $T$.  Now the sum 
$$ \sum_{t = t^{[3]}}^T \sum_i v_i \Delta_t x^Z_{it}$$
which by condition (a) of Lemma \ref{onethird} is at least $V^*/3$, can be split into three terms, some
of which may be empty:
\begin{itemize}
\item [(a)] $(\sum_{t = t^{[3]}}^{T} \Delta_t)(\sum_i v_i x_{i, t^{[3]}})$
\item [(b)] $(\sum_{t = t_1}^{T} \Delta_t)(\sum_{i \in A(1)} v_i)$
\item [(c)] $\sum_{t = t_2}^T \Delta t \sum_{i \in A(2)} v_i x_{i, t}$.
\end{itemize}
So the largest of these three terms has value at least $V^*/9$.  If it is (a) or (b)
then we have that a replicated-knapsack solution of value at least $V^*/9$ and
we are done.  If it is (c), then note that by (\ref{high1}) (and the pushed to the left
condition) 
\begin{eqnarray}
&& \sum_{i \in A(2)} w_i \le B_{T} - B_{t^{[3]}} \le B_{T} - B_{t_{1/2}} \le B_{t_{1/2}}, \label{packit}
\end{eqnarray}
where the inequalities follow because $t_{1/2} < t^{[3]}$ and $B_t$ is nondecreasing. It follows that we obtain a
feasible solution to $\IK$ by placing in the knapsack the set $A(2)$, at time $t_{1/2}$ (and no items added at any other
time).  This solution is feasible by (\ref{packit}), and its value is at least the quantity in (c) and so at least $V^*/9$.  This is the same as saying
that the replicated-knapsack solution at $t_{1/2}$ has value at least $V^*/9$. 
\end{proof}

\subsection*{Extension of Lemma \ref{halflemma} to general $0 < \kappa < 1$}
Here we outline how to extend Lemma \ref{halflemma} to obtain, for general $0 < \kappa < 1$, an approximation algorithm
with ratio $\min\{ \Omega(\frac{ 1 - \kappa }{\kappa}), \frac{1}{S(\kappa)} \}$.  The second case in the ``min'' 
corresponds to the case where $t^{[3]} < t_{\kappa}$, and it
follows by an analysis very similar to that of Lemma \ref{hardcase} (the reader will notice that in that proof, the
fact that $\kappa = 1/2$ was not actually used).  

Now we consider the case $t^{[3]} \ge t_{\kappa}$. First, we note that if $\kappa < 1/2$ is such that
an $(S, \kappa)$-split takes place at time $t$, then an $(S, 1/2)$-split takes place at $t$, as well (from (ii)) in 
Definition \ref{split}.  So we will assume $1/2 < \kappa < 1$.

Now suppose an $(S, \kappa)$-split takes place at time $t_{\kappa}$.  The idea is to split the time interval $[t_{\kappa}, T]$ into intervals $[t^0, t^1], \, [t^1, t^2], \ldots [t^{m-1}, t^m]$, where $t_{\kappa} = t^0 \le t^1 \le \ldots t^m = T$, $m = O( \kappa/(1 - \kappa))$, and the capacity increase experienced in each interval $[t^i, t^{i+1}]$ is $O(B_T/(1 - \kappa)$.  Then, at the boundary between successive intervals we apply an analysis similar to that used to prove Lemma \ref{easycase} to consider (1) two possible replicated-knapsacks (as in cases (a), (b) of the proof of Lemma \ref{easycase}), or (2) 
a solution similar to that in case (c) of Lemma \ref{easycase}, which in this case will use (by time $T$) total capacity
at most $(1 - \kappa)B_T$, and thus can be used to lower bound the replicated-knapsack solution at time $t_{\kappa}$.
Altogether, therefore, we obtain $O(3 \kappa/(1 - \kappa))$ replicated-knapsack solutions, and thus the best attains
an approximation factor

\subsection*{Proof of Lemma~\ref{NumDisjunctions}}
\begin{proof}
For a fixed $k \in \{1, \ldots, K\}$, we have that $\sigma^{k}_1 \leq \sigma^{k}_2 \leq \ldots \leq \sigma^{k}_T$. If $\sigma^{k}_T = m$, then since the $\sigma^{k}_i$s are integers, there are at most ${m+T-1 \choose m}$ feasible $T$-tuples $(\sigma^{k}_1, \sigma^{k}_2, \ldots , \sigma^{k}_T)$. Since $0 \leq m \leq J$, we have that ${m+T-1 \choose m} \leq (J+T)^J$. Consequently, there are at most $\sum_{m=1}^{J} {m+T-1 \choose m} \leq J(J+T)^J$ feasible $T$-tuples $(\sigma^{k}_1, \sigma^{k}_2, \ldots , \sigma^{k}_T)$. Thus, there are at most $(J(J+T)^J)^{K} =  O((1/\epsilon + T)^{O(\log(T/\epsilon )/\epsilon^2)})$ in the disjunctive procedure. \qed
\end{proof}

\subsection*{Proof of Theorem~\ref{Disjunctive procedure theorem}}

\begin{proof}
Let $\bar{x}$ be an optimal solution of $\max \{ \sum_{t=1}^{T}\sum_{i \in S^{h}} v'_{i}x_{i,t}: \ x \in Q^{\sigma,h}\}$. In order to prove the theorem, we will show the validity of the inequality (see Lemma \ref{inequality lemma 1})
\begin{align} \label{event rounding}\sum_{t =1}^{T}  \sum_{i=1}^{|S^{k, h}|} v'_ix^{\sigma,h}_{k_i,t} \geq (1-\epsilon) \sum_{t =1}^{T} \sum_{i=1}^{|S^{k, h}|} v'_i\bar{x}_{k_i,t}, \end{align} 
for every $S^{k,h}$ and that of (see Lemma \ref{inequality lemma 2})
\begin{align} \label{terminal rounding} \sum_{t=1}^{T} \sum_{i \in T^{h}} v'_ix^{\sigma,h}_{i,t} \geq \sum_{t=1}^{T}  \sum_{i \in T^{h}} v'_i\bar{x}_{i,t} - \epsilon v_h.\end{align}

The two inequalities imply:
\[\begin{split}
\sum_{t=1}^{T}\sum_{i \in S^{h}} v'_{i}x^{\sigma,h}_{i,t} &= 
v_h + \sum_{t =1}^{T}  \sum_{k=1}^{K} \sum_{i=1}^{|S^{k, h}|} v'_ix^{\sigma,h}_{k_i,t} + \sum_{t=1}^{T}  \sum_{i \in T^{h}} v'_ix^{\sigma,h}_{i,t} \\
&\geq v_h + (1-\epsilon)\sum_{t =1}^{T}  \sum_{k=1}^{K} \sum_{i=1}^{|S^{k, h}|} v'_i \bar{x}_{k_i,t} + \sum_{t=1}^{T}  \sum_{i \in T^{h}} v'_i\bar{x}_{i,t} - \epsilon v_h \\
&\geq \sum_{t=1}^{T}\sum_{i \in S^{h}} v'_{i}\bar{x}_{i,t}.
\end{split}
\]

\end{proof}

Next, we give proofs for equations (\ref{event rounding}) and (\ref{terminal rounding}). For every nonempty polyhedron $Q^{\sigma,h}$, we begin by showing (\ref{event rounding}) holds for every $S^{k,h}$ with the help of the following auxiliary LP.

\[
\begin{aligned}
\text{max} \quad
\sum_{t =1}^{T} \sum_{i=1}^{|S^{k,h}|} v'_{k_i}x_{k_i,t}  \notag \\
\text{s.t.} \quad
\sum_{i \in T^{h}}^{|S^{k,h}|}  w_{k_i}x_{k_i,t} \leq \sum_{i=1}^{|S^{k,h}|} & w_{k_i}\bar{x}_{k_i,t} \quad \forall t \notag \\
x_{k_1,t} = x_{k_2,t} = \ldots = x_{k_{|S^{k,h}|},t} = 0 &\ \quad \forall (k,t) \ \text{s.t.} \ \sigma^{k}_{t} = 0 \notag \\
x_{k_1,t} = x_{k_2,t} = \ldots = x_{k_{\sigma^{k}_{t}},t} = 1, & \\ \ x_{k_{\sigma^{k}_{t}} + 1,t} = \ldots = x_{k_{|S^{k,h}|},t} = 0 &\quad \forall (k,t) \ \text{s.t.} \  1 \leq \sigma^{k}_{t} < J \ \text{and} \ \sigma^{k}_t < |S^{k,h}| \notag \\
x_{k_1,t} = x_{k_2,t} = \ldots = x_{k_{\sigma^{k}_{t}},t} = 1, \ \sum_{i=1}^{|S^{k,h}|} x_{k_i, t} \geq \sigma^{k}_{t} & \quad \forall (k,t) \ \text{s.t.} \ \sigma^{k}_{t} = J \ \text{and} \ \sigma^{k}_t < |S^{k,h}| \notag \\
x_{k_1,t} = x_{k_2,t} = \ldots =  x_{k_{|S^{k,h}|},t} = 1\  &\quad \forall (k,t) \ \text{s.t.} \ \sigma^{k}_t \geq |S^{k,h}| \notag \\
x_{k_i, t-1} \leq x_{k_i, t} \quad &\forall k_i, \text{and} \ t = 2, 3, \ldots, T \notag \\
x_{k_i, t-1} \in [0,1] \quad &\forall k_i, t. \notag \\  \notag
\end{aligned}
\]
\begin{lemma} \label{number fractional variable}
For every $S^{k,h}$, there exists an optimal solution to the auxiliary LP that contains at most one fractional variable $x_{k_i,t}$ in each time period $t$. 
\end{lemma} 
\begin{proof}
The claim is true when $\sigma^{k}_{t} < J \ \text{and} \ \sigma^{k}_t < |S^{k,h}|$ or when $\sigma^{k}_t \geq |S^{k,h}|$ as there are no fractional variables in both cases. Hence, the only case left is when $\sigma^{k}_{t} = J \ \text{and} \ \sigma^{k}_t \geq |S^{k,h}|$.\\

Let $t^{\star}$ be the first (smallest) period in which we are in the case $\sigma^{k}_{t} = J \ \text{and} \ \sigma^{k}_t \geq |S^{k,h}|$.  Ignoring the precedence constraints for a moment, then the auxiliary LP can be broken up into $T-t^{\star}+1$ single period LPs of the following form, one for each $t \geq t^{\star}$. 

\begin{align}
LP_t = \text{max} \quad
\sum_{i=1}^{|S^{k,h}|} v'_ix_{k_i,t}  \notag \\
\text{s.t.} \quad
\sum_{i=1}^{|S^{k,h}|}  w_ix_{k_i,t} \leq \sum_{i=1}^{|S^{k,h}|} & w_i\bar{x}_{k_i,t} \label{capacity constraint 1} \\
x_{k_1,t} = x_{k_2,t} = \ldots = x_{k_{\sigma^{k}_{t}},t} = 1, \ \sum_{i=1}^{|S^{k,h}|} x_{k_i, t} \geq \sigma^{k}_{t} &\quad \forall (k,t) \ \text{s.t.} \ \sigma^{k}_{t} = J \ \text{and} \ \sigma^{k}_t < |S^{k,h}| \label{event constraint 1}\\
x_{k_i, t-1} \in [0,1] \quad &\forall k_i. \notag \\  \notag
\end{align}

Notice that in the modified instance of the problem, all items have the same value within a value class $S^{k,h}$. Hence, an optimal solution to $LP_t$ is simply to pack the items in the order of their weight, starting from the smallest weight item first. Moreover, notice that this set of optimal solutions satisfy the precedence constraints. 
\end{proof}

\begin{lemma} \label{inequality lemma 1}
Let $\bar{x}$ be an optimal solution to the optimization problem over a non-empty $Q^{\sigma, h}$ for some $\sigma \in \{0,\ldots, J\}^{TK}$ and $h \in S$, then there exists an integer feasible solution $x^{\sigma,h}$ to the auxiliary LP such that 
\[\sum_{t =1}^{T} \sum_{i=1}^{|S^{k,h}|} v'_ix^{\sigma,h}_{k_i,t} \geq (1-\epsilon) \sum_{t =1}^{T} \sum_{i=1}^{|S^{k,h}|} v'_i\bar{x}_{k_i,t}.\]
\end{lemma}
\begin{proof}
Without lost of generality, let $\hat{x}$ be the optimal solution to the auxiliary LP found using lemma \ref{number fractional variable}.  For time periods $t$ where $\sigma^{k}_{t} < J \ \text{and} \ \sigma^{k}_t < |S^{k,h}|$ or when $\sigma^{k}_t \geq |S^{k,h}|$, we don't need to round the variables $\hat{x}_{k_i,t}$ since they are already integral. Hence, we set $x^{\sigma,h}_{k_i,t} = \hat{x}_{k_i,t}$ for all the variables in this period, which implies that 

\[ \sum_{i=1}^{|S^{k,h}|} v'_ix^{\sigma,h}_{k_i,t} = \sum_{i=1}^{|S^{k,h}|} v'_i\hat{x}_{k_i,t} \geq  \sum_{i=1}^{|S^{k,h}|} v'_i\bar{x}_{k_i,t}\]

for such a period $t$. \\

For time periods $t$ where $\sigma^{k}_{t} = J \ \text{and} \ \sigma^{k}_t \geq |S^{k,h}|$, by lemma \ref{number fractional variable},  there is at most one fractional $\hat{x}_{k_i,t}$ in such a time period. Consequently, we round down this fractional variable while keeping others the same (or equivalently, setting $x^{\sigma,h}_{k_i,t} = 0$ for this variable and setting $x^{\sigma,h}_{k_i,t} = \hat{x}_{k_i,t}$). Since all the variables have the same value within a value class $S^{k,h}$, we have that 

\[ \frac{\sum_{i=1}^{|S^{k,h}|} v'_i\hat{x}_{k_i,t} - \sum_{i=1}^{|S^{k,h}|} v'_ix^{\sigma,h}_{k_i,t}}{\sum_{i=1}^{|S^{k,h}|} v'_i\hat{x}_{k_i,t}} \leq \frac{1}{J} < \epsilon,\]

which implies that 
\[ \sum_{i=1}^{|S^{k,h}|} v'_ix^{\sigma,h}_{k_i,t} \geq (1 -\epsilon)\sum_{i=1}^{|S^{k,h}|} v'_i\hat{x}_{k_i,t} \geq (1-\epsilon)  \sum_{i=1}^{|S^{k,h}|} v'_i\bar{x}_{k_i,t}.\]

Summing up the above inequalities over all time periods gives us the desired result.
\end{proof}

\begin{lemma} \label{inequality lemma 2}
Let $\bar{x}$ be an optimal solution to the optimization problem over a non-empty $Q^{\sigma, h}$ for some $\sigma \in \{0,\ldots, J\}^{TK}$ and $h \in S$, then there exists an integer feasible solution $x^{\sigma,h}$ to the auxiliary LP such that 
\[\sum_{t=1}^{T}  \sum_{i \in T^{h}} v'_ix^{\sigma,h}_{i,t} \geq \sum_{t=1}^{T}  \sum_{i \in T^{h}} v'_i\bar{x}_{i,t} - \epsilon v_h.\]
\end{lemma}
\begin{proof}
Consider the following auxiliary LP:

\begin{align}
\text{max} \quad
\sum_{t =1}^{T} \sum_{i \in T^{h}} v'_{i}x_{i,t}  \notag \\
\text{s.t.} \quad
\sum_{i \in T^{h}}  w_ix_{i,t} \leq \sum_{i \in T^{h}} & w_{i}\bar{x}_{i,t} \quad \forall t \notag \\
x_{i, t-1} \leq x_{i, t} \quad &\forall i \in T^{h}, \text{and} \ t = 2, 3, \ldots, T \notag \\
x_{i, t-1} \in [0,1] \quad &\forall i, t. \notag \\  \notag
\end{align}
An optimal solution of the LP above would be to greedily pack items in the order of non-increasing value to weight ratio. Let $\hat{x}$ be such an optimal solution, then it is clear that $\hat{x}$ has at most one fractional variable in each time period. We round down such a fractional variable in each time period to $0$ to obtain an integer solution $x^{\sigma,h}$. Consequently, we have that 

\[\sum_{t=1}^{T} \sum_{i \in T^{h}} v'_i\hat{x}_{i,t} - \sum_{t=1}^{T} \sum_{i \in T^{h}} v'_ix^{\sigma}_{i,t} \leq \frac{\epsilon  v_h}{T}\sum_{t=1}^{T} 1 =  \epsilon v_h,\]
where the first inequality follows from the fact that every item in $T^{h}$ has value weakly less than $\epsilon v_h / T$. Rearrange the terms gives us that:

\[\sum_{t=1}^{T} \sum_{i \in T^{h}} v'_ix^{\sigma}_{i,t} \geq \sum_{t=1}^{T} \sum_{i \in T^{h}} v'_i\hat{x}_{i,t}  - \epsilon v_h \geq \sum_{t=1}^{T} \sum_{i \in T^{h}}  v'_i\bar{x}_{i,t} - \epsilon v_h.\]
\end{proof}

\tiny   Fri.Nov.15.114345.201

\end{document}